\documentclass{llncs}

\usepackage[utf8]{inputenc}
\usepackage{graphicx} 
\usepackage{tikz}
\usepackage{color}
\usepackage{colortbl}
\usepackage{amsmath}
\usepackage{amssymb}
\usepackage{fixltx2e}
\usepackage{algorithm}
\usepackage{algorithmic}
\usepackage{xspace}

\definecolor{purple}{rgb}{0.7, 0, 1}
\definecolor{lightgray}{gray}{0.7}
\newcolumntype{G}{>{\columncolor{lightgray}}}

\title{Genome Halving by Block Interchange}
\author{Antoine Thomas \and A{\"i}da Ouangraoua \and Jean-St{\'e}phane Varr{\'e}}
\institute{LIFL, UMR 8022 CNRS, Universit\'e Lille 1
  \\ INRIA Lille, Villeneuve d'Ascq, France}

\newcommand{\breakpoint}{ \textsubscript{$\textcolor{black}{\blacktriangle}$} }

\newcommand{\fst}[1]{ \ensuremath{#1} }
\newcommand{\snd}[1]{ \ensuremath{\overline{#1}} }

\newcommand\aff[2]{\ensuremath{(\fst{#1}~~\fst{#2})}}
\newcommand\asf[2]{\ensuremath{(\snd{#1}~~\fst{#2})}}
\newcommand\afs[2]{\ensuremath{(\fst{#1}~~\snd{#2})}}
\newcommand\ass[2]{\ensuremath{(\snd{#1}~~\snd{#2})}}
\newcommand\paff[2]{\ensuremath{\fst{#1}~~\fst{#2}}}
\newcommand\pasf[2]{\ensuremath{\snd{#1}~~\fst{#2}}}
\newcommand\pafs[2]{\ensuremath{\fst{#1}~~\snd{#2}}}
\newcommand\pass[2]{\ensuremath{\snd{#1}~~\snd{#2}}}
\newcommand\oiff[2]{\ensuremath{]\fst{#1}~;~\fst{#2}[}}

\newcommand\oifs[2]{\ensuremath{]\fst{#1}~;~\snd{#2}[}}
\newcommand\oiss[2]{\ensuremath{]\snd{#1}~;~\snd{#2}[}}

\newcommand\cifs[2]{\ensuremath{[\fst{#1}~;~\snd{#2}]}}

\renewcommand{\NG}{\ensuremath{\mbox{\texttt{NG}}}}
\def\bi{\ensuremath{\mbox{BI}}}
\def\BI{\ensuremath{\mbox{BI}}}

\def\DCJ{\ensuremath{\mbox{\sl DCJ}}}

\def\etal{\textsl{et al.}\xspace}

\newlength{\arclength}
\newsavebox{\arcbox}

\begin{document}

\maketitle

\begin{abstract}
    We address the problem of finding the minimal number of block
    interchanges (exchange of two intervals) required to transform a
    duplicated linear genome into a tandem duplicated linear genome. We
    provide a formula for the distance as well as a polynomial time
    algorithm for the sorting problem.
\end{abstract}

\section{Introduction}
\label{sec:intro}

Genomic rearrangements are known to play a central role in the
evolutionary history of the species. Several operations act on the
genome, shaping the sequence of genes. A number of models to sort 
a genome into another have been studied: reversals,
transpositions and more recently Double-Cut-and-Join (DCJ). 
Another operation, called
\emph{block interchange}, consists in exchanging two intervals of a
genome. 

Block interchanges scenarios have been studied for the first time by
Christie \cite{Christie96}. He proposed a polynomial-time algorithm
for computing the distance between two linear chromosomes with unique
gene content. Lin \etal \cite{Lin05} proposed later a better algorithm.
Yancopoulos \etal \cite{Yancopoulos05} introduced the DCJ operation 
which consist in
cutting the genomes in two points and joining the four resulting
extremities in a different way. Interestingly, they noticed that a
block interchange can be simulated by two consecutive DCJs: an
excision followed by a reintegration. 

Another very important feature in genome evolution is that genomes
often undergo duplication events: both segmental and whole-genome
duplications. Genome duplication events are followed by other
rearrangements events which result in a scrambled genome. Genome
halving consists in finding the sequence of events that allow to go
back from the scrambled genome to the original duplicated one.

Genome halving has been studied under several models: reversals
\cite{Mabrouk98}, translocation/reversals \cite{Mabrouk03}, DCJ
\cite{Warren08}, breakpoints \cite{Tannier08}. Most of the results led
to polynomial time algorithms. Particularly, under the DCJ model,
Mixtacki \cite{Mixtacki08} gave some useful results and data structures. 
In this paper, we derive our results from those
results. Very recently, Kov{\'a}{\v c} \etal \cite{Kovac10} addressed 
the problem of
reincorporating the temporary circular chromosomes induced by DCJs
immediately after their creation considering genome halving. 
Although this problem is obviously related to the
problem we address, the aim and results are not the same. We are interested 
in linear genomes, not in multilinear ones,
and we focus on pure block interchange scenarios whereas Kov{\'a}{\v c}
\etal focused on scenarios made of reversals, translocation, fusion,
fissions along with block interchanges.

Section \ref{sec:pre} gives definitions. In Section \ref{sec:lb}, we
first give a lower bound on the distance with helpful properties for
the rest of the paper. In Section \ref{sec:dist}, we prove the
analytical formula for the distance. We conclude in Section
\ref{sec:scenario} with a quadratic time and space algorithm to
obtain a parsimonious scenario.

\section{Preliminaries: duplicated genomes, rearrangement, genome halving problems}
\label{sec:pre}

In this section we give the main definitions and notations used in the paper.

\subsection*{Duplicated Genomes}

A genome is composed of genomic markers organized in linear or circular 
chromosomes. A linear chromosome is represented by an ordered sequence of 
unsigned integers, each standing for a marker, surrounded by two abstract 
markers $\circ$  at each end indicating the telomeres. A circular chromosome 
is represented by a circularly ordered sequence of unsigned integers 
representing markers. For example,  
$(\fst{1}~~\fst{2}~~\fst{3}) ~
(\circ~~\fst{4}~~\fst{5}~~\fst{6}~~\fst{7}~~\circ)$ is a genome
constituted of one circular and one linear chromosome. 

\begin{definition}
A \emph{rearranged duplicated genome} is a genome in which each marker appears 
twice. 
\end{definition}

In a rearranged duplicated genome, two copies of a same marker are called paralogs. We distinguish paralogs by denoting one marker by $\fst{x}$ and its paralog by $\snd{x}$. By convention $\snd{\snd{x}}=x$.
For example, the following genome is a rearranged duplicated genome: 
$(\circ~~\fst{1}~~\snd{1}~~\fst{3}~~\fst{2}~~\fst{4}~~\fst{5}~~\fst{6}~~\snd{6}~~\fst{7}~~\snd{3}~~\fst{8}~~\snd{2}~~\snd{4}~~\snd{5}~~\fst{9}~~\snd{8}~~\snd{7}~~\snd{9}~~\circ )$.

An \emph{adjacency} in a genome is a pair of consecutive markers. 
For example, the genome $(\circ~~\fst{1}~~\fst{2}~~\circ)~ (\fst{3}~~\fst{4}~~\fst{5})$ has six adjacencies, $(\circ ~~\fst{1}),~(\fst{1}~~\fst{2}) ,~ (\fst{2}~~\circ)$, and $(\fst{3}~~\fst{4}) ,~ (\fst{4}~~\fst{5}) ,~ (\fst{5}~~\fst{3})$.
The linear 
or circular order of the markers in a chromosome naturally induces an order 
on the adjacencies that we denote by $<$. For example in the previous genome 
the order induced on the adjacencies is:
 $(\circ ~~\fst{1}) < (\fst{1}~~\fst{2}) < (\fst{2}~~\circ)$, and $(\fst{3}~~\fst{4}) < (\fst{4}~~\fst{5}) < (\fst{5}~~\fst{3}) < (\fst{3}~~\fst{4})$.

A \emph{double-adjacency} in a genome $G$ is an adjacency $\aff{a}{b}$
such that  $\ass{a}{b}$ is an adjacency of $G$ as well. Note that a genome always has an even number of double-adjacencies.
For example, the four double-adjacencies in the following genome are indicated by dots :
$$G = (\circ~~\fst{1}~~\snd{1}~~\fst{3}~~\fst{2}~\cdot~\fst{4}~\cdot~\fst{5}~~\fst{6}~~\snd{6}~~\fst{7}~~\snd{3}~~\fst{8}~~\snd{2}~\cdot~\snd{4}~\cdot~\snd{5}~~\fst{9}~~\snd{8}~~\snd{7}~~\snd{9}~~\circ )$$

A consecutive sequence of double-adjacencies can be rewritten as a single marker; this process is called \emph{reduction}. For example, genome $G$ can be reduced by rewritting $\fst{2}~.~\fst{4}~.~\fst{5}$ and $\snd{2}~.~\snd{4}~.~\snd{5}$ as $\fst{10}$ and $\snd{10}$, yielding the following genome:
$$G^r = (\circ ~~\fst{1}~~\snd{1}~~\fst{3}~~\fst{10}~~\fst{6}~~\snd{6}~~\fst{7}~~\snd{3}~~\fst{8}~~\snd{10}~~\fst{9}~~\snd{8}~~\snd{7}~~\snd{9}~~\circ )$$

\begin{definition}
A \emph{tandem-duplicated genome} is a rearranged duplicated genome which can be reduced to a genome of the form $(\circ~~\fst{x}~~\snd{x}~~\circ )$.
\end{definition}

In other words, a tandem-duplicated genome is composed of a single linear chromosome where all adjacencies, except the two containing the marker $\circ$ and the central adjacency, are double-adjacencies.
For example, the genome $(\circ~~\fst{1}~\cdot~\fst{2}~\cdot~\fst{3}~\cdot~\fst{4}~~\snd{1}~\cdot~\snd{2}~\cdot~\snd{3}~\cdot~\snd{4}~~\circ )$ is a tandem-duplicated genome that can be reduced to 
$(\circ~~\fst{5}~~\snd{5}~~\circ )$
by rewritting $\fst{1}~\cdot~\fst{2}~\cdot~\fst{3}~\cdot~\fst{4}$ and $\snd{1}~\cdot~\snd{2}~\cdot~\snd{3}~\cdot~\snd{4}$ as  $\fst{5}$ and $\snd{5}$.


\begin{definition}
A \emph{perfectly duplicated genome} is a rearranged duplicated genome such that each adjacency is a double-adjacency. 
\end{definition}

For example, the genome $(1~~{2}~~{3}~~4~~\snd{1}~~\snd{2}~~\snd{3}~~\snd{4})$ is a perfectly duplicated genome.

\subsection*{Rearrangements}
\label{sec:rearrangement}

A rearrangement operation on a given genome cuts a set of adjacencies of the 
genome called \emph{breakpoints} and forms new adjacencies with the exposed 
extremities, while altering no other adjacency. In the sequel, the adjacencies 
cut by a rearrangement operation are indicated in the genome by the symbol $\breakpoint$. 
  
An \emph{interval} in a genome is a set of markers that appear consecutively 
in the genome. Given two different adjacencies $(a~~b)$ and $(c~~d)$ in a genome 
$A$ such that $(a~~b) < (c~~d)$, $[b~;~c]$ denotes the interval of $A$ beginning 
with marker $b$ and ending with marker $c$.

In this paper, we consider two types of rearrangement operations called \emph{block interchange (BI)} and  \emph{double-cut-and-join (DCJ)}.

A \emph{block interchange} (\bi) on a genome $G$ is a rearrangement 
operation that acts on four adjacencies in $G$, 
$(a~~b) <  (c~~d) \leq (u~~v) < (x~~y)$ such that the intervals $[b~;~c]$ and 
$[v~;~x]$ do not overlap, 
swapping the intervals $[b~;~c]$ and $[v~;~x]$.
For example, the following block interchange acting on adjacencies 
$(\snd{1}~~\fst{2})< (\fst{6}~~\snd{6})< (\snd{3}~~\fst{8})< (\snd{8}~~\snd{7})$
consists in swapping the intervals $[\fst{2},\fst{6}]$ and $[\fst{8},\snd{8}]$. 
\begin{center}
$(\circ~~\fst{1}~~\snd{1}~\breakpoint~\mathbf{\fst{2}~~\fst{3}~~\snd{2}~~\fst{4}~~\fst{5}~~\fst{6}} ~\breakpoint~\snd{6}~~\fst{7}~~\snd{3} ~\breakpoint~~\mathbf{\fst{8}~~\snd{4}~~\fst{9}~~\snd{5}~~\snd{8}}~ \breakpoint~\snd{7}~~\snd{9}~~\circ)$

$\downarrow$

$(\circ~~\fst{1}~~\snd{1}~~\mathbf{\fst{8}~~\snd{4}~~\fst{9}~~\snd{5}~~\snd{8}}~~\snd{6}~~\fst{7}~~\snd{3}~~\mathbf{\fst{2}~~\fst{3}~~\snd{2}~~\fst{4}~~\fst{5}~~\fst{6}}~~\snd{7}~~\snd{9}~~\circ)$
\end{center}

A \emph{double-cut-and-join} (DCJ) operation on a genome $G$ cuts two different 
adjacencies in $G$ and glues pairs of the four exposed extremities to form two 
new adjacencies. 
Here, we focus on two types of DCJ operations called \emph{excision} and \emph{integration}.

An \emph{excision} is a DCJ operation acting on a single chromosome by extracting an interval from it, making this interval a circular chromosome, and making the remainder a single chromosome (1 join). For example, the following excision extracts the circular chromosome $(\fst{2}~~\fst{3}~~\fst{4})$:
$$(\circ~\fst{1}~\breakpoint~\mathbf{\fst{2}~~\fst{3}~~\fst{4}}~\breakpoint~\fst{5}~~\fst{6}~\circ) \rightarrow (\mathbf{\fst{2}~~\fst{3}~~\fst{4}}) (\circ \fst{1}~~\fst{5}~~\fst{6}~\circ)$$ 

An \emph{integration} is the inverse of an excision; it is a DCJ operation that acts on two chromosomes, one being a circular chromosome, to produce a single chromosome.  For example, the following operation is an integration of the circular chromosome $(\fst{2}~~\fst{3}~~\fst{4})$:
$$\mathbf{(\fst{2}~\breakpoint~\fst{3}~~\fst{4}}) (\circ \fst{1}~~\fst{5}~~\fst{6}~\breakpoint~\circ) \rightarrow (\circ~\fst{1}~~\fst{5}~~\fst{6}~~\mathbf{\fst{3}~~\fst{4}~~\fst{2}}~\circ)$$

We now give an obvious, but very useful, property linking BI operations to DCJ operations.

\begin{property}
\label{1BIto2DCJ}
A single BI operation on a linear chromosome is equivalent to two DCJ operations: an excision followed by an integration.
\end{property}

\begin{proof}
Let $(\circ~~1~~U~~2~~V~~3~~\circ)$ be a genome, $U$ and $V$ the two intervals 
that are to be swapped by a block interchange operation, $1$ $2$ and $3$ the 
intervals constituting the rest of the genome (note that each of them may be 
empty). 

The first DCJ operation is the excision that produces the adjacency
$(1~~V)$ by extracting and circularizing the interval $[U~;~2]$: 
$$(\circ~~1~\breakpoint~U~~2~\breakpoint~V~~3~~\circ) \rightarrow
(\circ~~1~~V~~3~~\circ)(U~~2~)$$

The second DCJ operation is the integration that produces the adjacency
$(U~~3)$ by reintegrating the circular chromosome $(U~~2)$ in the
appropriate way:
$$(\circ~~1~~V~\breakpoint~3~~\circ)(U~~2~\breakpoint) \rightarrow (\circ~~1~~V~~2~~U~~3~~\circ)$$\qed
\end{proof}

A \emph{rearrangement scenario} between two genomes $A$ and $B$ is a sequence of rearrangement operations allowing to transform $A$ into $B$.

\begin{definition}
A \emph{BI (resp. DCJ) scenario} is a rearrangement scenario composed of BI (resp. DCJ) operations.
\end{definition}


The length of a rearrangement scenario is the number of rearrangement operations composing the scenario. 
\begin{definition}
The \emph{BI (resp. DCJ) distance} between two genomes $A$ and $B$, denoted by $d_{BI}(A,B)$ (resp. $d_{DCJ}(A,B)$), is the minimal length of a  BI (resp. DCJ) scenario between $A$ and $B$.
\end{definition}

\subsection*{Genome Halving}
\label{sec:halving}

We now state the genome halving problem considered in this paper.

\begin{definition}
Given a rearranged duplicated genome $G$ composed of a single linear chromosome, the \emph{BI halving problem} consists in finding a tandem-duplicated genome $H$ such that the BI distance between $G$ and $H$ is minimal.
\end{definition}

In order to solve the BI halving problem, we use some results on the \emph{DCJ halving problem} that were stated in \cite{Mixtacki08} as a starting point. Unlike the BI halving problem, the aim of the DCJ halving problem is to find a perfectly duplicated genome instead of a tandem-duplicated genome.

\begin{definition}[\cite{Mixtacki08}]
Given a rearranged duplicated genome $G$, the \emph{DCJ genome halving problem} 
consists in finding a perfectly duplicated genome $H$ such that the DCJ distance
between $G$ and $H$ is minimal.
\end{definition}

The BI and DCJ genome halving problems lead to two definitions of \emph{halving distances}: the \emph{BI halving distance} (resp. \emph{DCJ halving distance}) of a rearranged duplicated genome $G$ is the minimum BI (resp. DCJ) distance between $G$ and any tandem-duplicated genome (resp. any perfectly duplicated genome) ; we denote it by $d^t_{BI}(G)$ (resp. $d^p_{DCJ}(G)$).\\

\section{Lowerbound for the BI halving distance}
\label{sec:lb}

In this section we give a lowerbound on the BI halving distance of a rearranged duplicated genome. We use a data structure representing the genome called the \emph{natural graph} introduced in \cite{Mixtacki08}.

\begin{definition}
The natural graph of a rearranged duplicated genome $G$, denoted by $\NG(G)$, is the graph whose vertices are the \emph{adjacencies} of $G$, and for any marker $u$ there is one edge between  $\aff{u}{v}$ and  $\asf{u}{w}$, and  one edge between  $\aff{x}{u}$ and  $\afs{y}{u}$.
\end{definition}

Note that the number of edges in the natural graph of a genome $G$ containing $n$ distinct markers, each one present in two copies, is always $2n$. Moreover, since every vertex has degree one or two, then the natural graph consists only of cycles and paths. For example, the natural graph of genome $G = (\circ~~ \fst{1}~~\snd{2}~~\snd{1}~~\snd{4}~~\fst{3}~~\fst{4}~~\snd{3}~~\fst{2}~~\circ)$ is depicted in Fig. \ref{fig:NGdef}.

\begin{figure}[htbp]
    \centering
\begin{tikzpicture}
    \matrix[row sep=1mm,column sep=2mm,ampersand replacement=\&] {
      \node[draw,circle] (LA1)  {\paff{\circ}{1}}; \&
      \node[draw,circle] (B2A2) {\pass{2}{1}}; \&
      \node[draw,circle] (B1R)  {\paff{2}{\circ}}; \&
      \node[draw,circle] (A1B2) {\pafs{1}{2}}; \&
      \node[draw,circle] (A2D2) {\pass{1}{4}}; \&
      \node[draw,circle] (C1D1) {\paff{3}{4}}; \&
      \node[draw,circle] (C2B1) {\pasf{3}{2}}; \&
      \node[draw,circle] (D2C1) {\pasf{4}{3}}; \&
      \node[draw,circle] (D1C2) {\pafs{4}{3}}; \\
    };

    \draw (A1B2) to [out=60,in=120] (A2D2);
    \draw (A1B2) to [out=-60,in=-120] (C2B1);

    \draw (B2A2) to [out=0,in=180] (B1R);
    \draw (B2A2) to [out=180,in=0] (LA1);

    \draw (A2D2) to [out=60,in=120] (C1D1);

    \draw (D2C1) to [out=60,in=120] (D1C2);
    \draw (D2C1) to [out=-60,in=-120] (D1C2);

    \draw (C1D1) to [out=60,in=120] (C2B1);

\end{tikzpicture}

\caption{The natural graph of genome $G = (\circ~~ \fst{1}~~\snd{2}~~\snd{1}~~\underline{\snd{4}~~\fst{3}}~~\underline{\fst{4}~~\snd{3}}~~\fst{2}~~\circ)$ ; it is composed of one path and two cycles.}
\label{fig:NGdef}
\end{figure}



\begin{definition}
    Given an integer $k$, a \emph{$k-$cycle} (resp. \emph{$k-$path}) in the 
    natural graph of a rearranged duplicated genome is a cycle (resp. path) 
    that contains $k$ edges. If $k$ is even, the cycle (resp. path) is called 
   \emph{even}, and \emph{odd} otherwise.
\end{definition}

\def\EC{\ensuremath{\mbox{EC}}}
\def\OP{\ensuremath{\mbox{OP}}}

Based on the natural graph, a formula for the DCJ halving distance was given in \cite{Mixtacki08}. Given a rearranged duplicated genome $G$ such that the number of even cycles and the number of odd paths in $\NG(G)$
are respectively denoted by $\EC$  and $\OP$, the  DCJ halving distance of $G$ is:
    $$d^p_{DCJ}(G) = n - \EC - \left\lfloor \frac{\OP}{2} \right\rfloor $$

In the case of the BI halving distance, some peculiar properties of the natural graph need to be stated, allowing to simplify the formula of the DCJ halving distance, and leading to a lowerbound on the BI halving distance.

In the following properties, we assume that $G$ is a genome composed of a 
single linear chromosome containing $n$ distinct markers, each one present 
in two copies in $G$. 


\begin{property}
The natural graph $\NG(G)$ contains only even cycles and paths: 
\label{prop:EC}
\label{prop:1EP}
\begin{enumerate}
\item All cycles in the natural graph $\NG(G)$ are even.
\item The natural graph $\NG(G)$ contains only one path, and this path is even.
\end{enumerate}
\end{property}

\begin{proof}
    First, if \aff{a}{x} is a vertex of the
    graph that belongs to a cycle $C$, then there exists an edge between
    \aff{a}{x} and a vertex \asf{a}{y}. These two
    adjacencies are the only two containing a copy of the marker $a$ at the
    first position. So, if we consider the set of all the first markers in
    all adjacencies contained in the cycle $C$, then each marker in this set is
    present exactly twice. Therefore, the cycle $C$ is an even cycle.
  
    Secondly, the graph contains exactly two vertices (adjacencies) containing the
    marker $\circ$ which are both necessarily ends of a path in $\NG(G)$. Thus
   there can be only one path in the graph. Since the number of edges in the 
   graph is even and all cycles are even, then the single path is also even. \qed
\end{proof}

We now give a lowerbound on the minimum length of DCJ scenario transforming 
$G$ into a tandem-duplicated genome.

\begin{lemma}
    Let $d^t_{DCJ}(G)$  be the minimum DCJ distance between $G$ and any 
   tandem-duplicated genome. If $\NG(G)$ contains $C$ cycles then a 
   lowerbound on $d^t_{DCJ}(G)$ is given by:
    $$d^t_{DCJ}(G) \geq n - C - 1$$
\label{DCJdistLB}
\end{lemma}

\begin{proof}
    First, since all cycles  of $\NG(G)$ are even and $\NG(G)$ contains no odd 
    path, 
   then, from the DCJ halving distance formula, the DCJ halving distance of 
   $G$   is $d^p_{DCJ}(G) = n - C$.

    Now, since any  tandem-duplicated genome can be transformed into 
    a perfectly duplicated genome with one DCJ, then $d^t_{DCJ} + 1
    \geq d^p_{DCJ}$. Therefore, we have $ d^t_{DCJ} \geq d^p_{DCJ} - 1 \geq
    n - C - 1$.  \qed
\end{proof}

We are now ready to state a lowerbound on the BI halving distance of a rearranged duplicated genome $G$.

\begin{theorem}
    If $\NG(G)$ contains $C$ cycles, then a lowerbound on the BI halving distance
    is given by: 
$$d^t_{BI}(G) \geq \left \lfloor \frac{n - C}{2} \right \rfloor$$
\end{theorem}

\begin{proof}

We denote by $\ell(S)$ the length of a rearrangement scenario $S$.
   Let $S_{BI}$ be a BI scenario transforming $G$ into a
   tandem-duplicated genome.
    From property \ref{1BIto2DCJ}, we have that $S_{BI}$ is equivalent to
    a DCJ scenario $S_{DCJ}$ such that  $\ell(S_{DCJ})=2*\ell(S_{BI})$. 
    Now, suppose that $\ell(S_{BI}) < \lfloor \frac{n - C}{2}
    \rfloor$, then $\ell(S_{BI}) \leq \lfloor \frac{n - C}{2}
    \rfloor - 1 \leq \lceil \frac{n - C - 1}{2} \rceil - 1$.
    
    This implies $\ell(S_{DCJ}) \leq 2\lceil \frac{n - C - 1}{2}
    \rceil - 2 \leq n - C - 2 < n - C - 1$. Thus, from Lemma 
   \ref{DCJdistLB} we have $\ell(S_{DCJ}) < d^t_{DCJ}$ 
    which contradicts the fact that $d^t_{DCJ}$ is the minimal number
    of DCJ
    operations required to transform $G$ into a tandem-duplicated genome.

    In conclusion, we always have $d^t_{BI}(G) \geq \lfloor \frac{n - C}{2} \rfloor$.  \qed
\end{proof}

\section{Formula for the BI halving distance}
\label{sec:dist}

In this section, we show that the BI halving distance of a rearranged duplicated genome $G$ with $n$ distinct markers such that  $\NG(G)$ contains $C$ cycles is exactly:

$$d^t_{BI}(G) = \left \lfloor \frac{n - C}{2} \right \rfloor$$

In other words, we show that enforcing the constraint that 2 consecutive DCJ have to be equivalent to a BI doesn't change the distance (even though it obviously restricts the DCJ that can be performed at each step of the scenario). 







In
the following, $G$ denotes a rearranged duplicated genome $G$
constisting in a single linear chromosome with $n$ distinct markers
after the reduction process, and such that $\NG(G)$ contains $C$ cycles.
We begin by recalling some useful definitions and properties of the DCJ
operations that allow to decrease the DCJ halving distance by $1$ in the resulting genome. 

\begin{definition}
A DCJ operation on $G$ producing genome $G'$ is \emph{sorting} if it decreases the DCJ halving distance by $1$: $d^p_{DCJ}(G') = d^p_{DCJ}(G)-1 = n-C-1$.
\end{definition}

Since the number of distinct markers $G'$ is $n$ and $d^p_{DCJ}(G') = n-C-1$, then $\NG(G')$ contains $C+1$ cycles. In other words, a DCJ operation is sorting if it increases the number of cycles in  $\NG(G)$ by $1$.

Given $(\fst{u}~~\fst{v})$ an adjacency of $G$ that is not a double-adjacency,
we denote by $\DCJ(\fst{u}~~\fst{v})$ the DCJ operation that cuts adjacencies $(\snd{u}~~\fst{x})$ and  $(\fst{y}~~\snd{v})$ to form adjacencies $(\snd{u}~~\snd{v})$ and  $(\fst{y}~~\fst{x})$, making $(\fst{u}~~\fst{v})$ a double-adjacency. 

\begin{property}
Let  $(\fst{u}~~\fst{v})$ be an adjacency of $G$ that is not a double-adjacency,
 $\DCJ(\fst{u}~~\fst{v})$ is a sorting DCJ operation.
\end{property}
\begin{proof}
$\DCJ(\fst{u}~~\fst{v})$ increases the number of cycles in  $\NG(G)$ by $1$, by creating a new cycle composed of adjacencies  $(\fst{u}~~\fst{v})$ and  $(\snd{u}~~\snd{v})$. \qed
\end{proof}













\begin{figure}[!h]
\centering
\begin{tikzpicture}
    \matrix[row sep=1mm,column sep=8mm,ampersand replacement=\&] {

\node (L) {$(~\circ$}; \&
\node (f2) {$\fst{2}$}; \&
\node (f1) {$\fst{1}$}; \&
\node (s2) {$\snd{2}$}; \&
\node (f3) {$\fst{3}$}; \&
\node (s1) {$\snd{1}$}; \&
\node (s3) {$\snd{3}$}; \&
\node (R) {$\circ~)$}; \\
};
\draw[very thin] (0.1,-0.3) -- ++(1,0) -- ++(0,0.6) -- ++(-1,0) --
++(0,-0.6) -- ++(0,0.6) -- ++(1,0) -- ++(0.5,1) node[above]
{$I\aff{2}{1} = \oiss{2}{1}$};
\draw[very thick] (-3.4,-0.4) -- ++(5.8,0) -- ++(0,0.8) -- ++(-5.8,0)
-- ++(0,-0.8) -- ++(-0.3,-0.7) node[below] {$I\afs{1}{2} = \cifs{2}{1}$};
\draw[very thin] (-2.6,-0.6) -- ++(5.2,0) -- ++(0,1.2) -- ++(-5.2,0) --
++(0,-1.2) -- ++(0,1.2) -- ++(-0.5,1) node[above] {$I\asf{2}{3} = \oifs{2}{3}$};
\draw[very thick] (-2.3,-0.5) -- ++(5.5,0) -- ++(0,1) -- ++(-5.5,0) --
++(0,-1) -- ++(5.5,0) -- ++(0.3,-0.7) node[below] {$I\afs{3}{1} = \cifs{1}{3}$};
\draw[very thin] (-1.1,-0.3) -- ++(1,0) -- ++(0,0.6) -- ++(-1,0) --
++(0,-0.6) -- ++(0,0.6) -- ++(-0.3,0.7) node[above] {$I\ass{1}{3} = \oiff{1}{3}$};

\end{tikzpicture}
\\

\caption{ $\mathcal{I}(G) ~ = ~ \left\{ ~~ \oiss{2}{1} ~,~~
      \mathbf{\cifs{2}{1}} ~,~~ \oifs{2}{3} ~,~~ \mathbf{\cifs{1}{3}}
      ~,~~ \oiff{1}{3} ~~ \right\}$, the set of intervals of $G =
  (\circ~~\fst{2}~~\fst{1}~~\snd{2}~~\fst{3}~~\snd{1}~~\snd{3}~~\circ
  )$ depicted as boxes. The two boxes with thick lines represent two 
  overlapping intervals of $\mathcal{I}(G)$ inducing a \BI\xspace which 
  exchanges $\fst{2}$ and $\snd{3}$.}
    \label{fig:intervaldef}
\end{figure}

\begin{definition}
Let $(\fst{u}~~\fst{v})$,  $(\snd{u}~~\fst{x})$, and  $(\fst{y}~~\snd{v})$ be adjacencies of $G$. The \emph{interval} of the adjacency  $(\fst{u}~~\fst{v})$, denoted by $I(\fst{u}~~\fst{v})$ is either:
\begin{itemize}
\item the interval $[\snd{x}~;~\snd{y}]$ if $(\snd{u}~~\fst{x}) < (\fst{y}~~\snd{v})$.  In this case, we denote it by $]\snd{u}~;~\snd{v}[$, or
\item the interval $[\snd{v}~;~\snd{u}]$ if $(\fst{y}~~\snd{v}) < (\snd{u}~~\fst{x})$.
\end{itemize}
\end{definition}

For example, the intervals of the adjacencies in genome $(\circ~~\fst{2}~~\fst{1}~~\snd{2}~~\fst{3}~~\snd{1}~~\snd{3}~~\circ )$ are depicted in Fig \ref{fig:intervaldef}.
Note that, given an adjacency  $(\fst{u}~~\fst{v})$ of $G$,  if  $(\fst{u}~~\fst{v})$ is a double-adjacency then the interval $I(\fst{u}~~\fst{v})$ is empty, otherwise  $\DCJ(\fst{u}~~\fst{v})$ is the excision operation that extracts the interval $I(\fst{u}~~\fst{v})$ to make it circular, thus producing the adjacency $(\snd{u}~~\snd{v})$.


Two intervals  $I\aff{a}{b}$ and $I\aff{x}{y}$ are said \emph{overlapping} if
their intersection is non-empty, and none of the intervals is included in the 
other.
It is easy to see, following Property \ref{1BIto2DCJ}, that given two
adjacencies $\aff{a}{b}$ and $\aff{x}{y}$ of $G$ such that
$I\aff{a}{b}$ and $I\aff{x}{y}$ are non-empty intervals, the
successive application of $\DCJ\aff{a}{b}$ and $\DCJ\aff{x}{y}$ is
equivalent  to a BI operation if and only if $I\aff{a}{b}$ and
$I\aff{x}{y}$ are overlapping. Note that in this case neither
$\aff{a}{b}$, nor  $\aff{x}{y}$ can be double-adjacencies in $G$ since
their  intervals are non-empty. Figure \ref{fig:intervaldef} shows an
example of two overlapping intervals.

The following property states precisely in which case the successive application of $\DCJ\aff{a}{b}$ and $\DCJ\aff{x}{y}$ decreases the DCJ halving distance by $2$, meaning that both DCJ operations are sorting.

\begin{property}
Given two adjacencies $\aff{a}{b}$ and $\aff{x}{y}$ of $G$, such that 
$I\aff{a}{b}$ and $I\aff{x}{y}$ are overlapping, the successive application of $\DCJ\aff{a}{b}$ and $\DCJ\aff{x}{y}$ decreases the DCJ halving distance by $2$ if and only if  $x \neq \snd{a}$ and $y \neq \snd{b}$.
\end{property}
\begin{proof}
If $x \neq \snd{a}$ and $y \neq \snd{b}$, then the successive application of $\DCJ\aff{a}{b}$ and $\DCJ\aff{x}{y}$ increases the number of cycles in $\NG(G)$ by $2$, by creating two new 2-cycles. Otherwise, $\DCJ\aff{a}{b}$ first creates a new cycle that is then destroyed by  $\DCJ\aff{x}{y}$.
\qed
\end{proof}

We denote by $\mathcal{I}(G)$, the set of intervals of all the adjacencies of $G$ that do not contain marker $\circ$.

\begin{remark}
Note that, if $G$ contains $n$ distinct markers, then there are $2n-1$ adjacencies in $G$ that do not contain marker $\circ$,  defining $2n-1$ intervals in  $\mathcal{I}(G)$.
\label{maxEdges}
\end{remark}

\begin{definition}
Two  intervals $I\aff{a}{b}$ and $I\aff{x}{y}$ of $\mathcal{I}(G)$ are said \emph{compatible} if they are overlapping and  $x \neq \snd{a}$ and $y \neq \snd{b}$.
\end{definition}


In the following, we prove the BI halving distance formula by showing that if genome $G$ contains more than three distinct markers, $n~>~3$,  then there exist two compatible intervals in $\mathcal{I}(G)$, and if $n=2$ or $n=3$ then $d^t_{BI}(G)=1$ and $2 \leq d^p_{DCJ}(G) \leq 3$.
This means that there exists a \BI ~halving scenario $S$ such that all \BI ~operations in $S$, possibly excluding the last one, are equivalent to two successive sorting DCJ operations.




From now on, until the end of the section,
$\aff{a}{b}$ is an adjacency of $G$ that is not a double-adjacency, $A$ is a genome consisting in a linear chromosome $\mathfrak{L}$ and a circular chromosome $\mathfrak{C}$, obtained by applying the \emph{sorting DCJ}, $\DCJ\aff{a}{b}$, on $G$.

If there exists an interval $I\aff{x}{y}$ in $\mathcal{I}(G)$ compatible with  $I\aff{a}{b}$, then applying $\DCJ\aff{x}{y}$ on $A$ consists in the integration of the circular chromosome $\mathfrak{C}$ into the linear chromosome $\mathfrak{L}$ such that the adjacency $\ass{x}{y}$ is formed.
Such an \emph{integration} can only be performed by cutting an adjacency $\asf{x}{u}$ in $\mathfrak{C}$ and an adjacency $\afs{v}{y}$ in $\mathfrak{L}$ (or inversely) to produce adjacencies $\ass{x}{y}$ and $\aff{v}{u}$.  This means that there must be an adjacency \aff{x}{y} in either $\mathfrak{C}$ or $\mathfrak{L}$ such that $\snd{x}$ is in $\mathfrak{C}$ and $\snd{y}$ in $\mathfrak{L}$ or inversely.
Hence, we have the following property :



\begin{property}
$\mathfrak{C}$ \emph{cannot} be reintegrated into $\mathfrak{L}$ by
applying a sorting DCJ,  $\DCJ\aff{x}{y}$, on $A$ if and only if either:

\begin{itemize} 
\item[(1)] for any adjacency $\aff{x}{y}$ in $\mathfrak{C}$
  (resp. $\mathfrak{L}$), markers $\snd{x}$ and $\snd{y}$ are in
  $\mathfrak{L}$ (resp. $\mathfrak{C}$), 
  or

\item[(2)] for any adjacency $\aff{x}{y}$ in $\mathfrak{C}$ (resp. $\mathfrak{L}$), markers $\snd{x}$ and $\snd{y}$ are also in $\mathfrak{C}$ (resp. $\mathfrak{L}$). 
\end{itemize}

\label{formuleENR}
\end{property}
\begin{proof}
If there exists no  adjacency $\aff{x}{y}$ in $A$ such that $\snd{x}$ is in $\mathfrak{C}$ and $\snd{y}$ in $\mathfrak{L}$ or inversely, then $A$ necessarily satisfies either $(1)$, or $(2)$. \qed
\end{proof}






\begin{definition}
An interval $I\aff{a}{b}$ in  $\mathcal{I}(G)$ is called
\emph{interval of type 1} (resp. \emph{interval of type 2}) if $\DCJ\aff{a}{b}$ produces a genome $A$ satisfying configuration $(1)$  (resp.  configuration $(2)$) described in Property \ref{formuleENR}.
\end{definition}


For example, in genome  $(\circ~~\fst{2}~~\fst{1}~~\snd{1}~~\fst{3}~~\snd{2}~~\snd{3}~~\circ)$,  $I\asf{1}{3}$ is of type 1 as  $\DCJ\asf{1}{3}$ produces genome 
$(\circ~~\fst{2}~~\fst{1}~~\snd{3}~~\circ) ~ (\snd{1}~~\fst{3}~~\snd{2})$ ;  $I\ass{2}{3}$ is of type 2 as  $\DCJ\ass{2}{3}$ produces genome  $(\circ~~\fst{2}~~\fst{3}~~\snd{2}~~\snd{3}~~\circ) ~ (\fst{1}~~\snd{1})$.

Now we give the maximum numbers of intervals of type 1 and type 2 that can be contained in genome $G$.

\begin{lemma}
The maximum number of intervals of type 1 in $\mathcal{I}(G)$ is 2.
\label{maxType1}
\end{lemma}

\begin{proof}
First, note that there cannot be two intervals $I$ and $J$ of $\mathcal{I}(G)$ 
such that $I \neq J$, and both  $I$ and $J$ are of type 1.
Now, if $I$ is an interval of type 1, there can be at most two different 
adjacencies $\aff{x}{y}$ and $\aff{u}{v}$ such that 
$I\aff{x}{y} = I\aff{u}{v} = I$. In this case $G$ necessarily has a chromosome of the form $(\ldots ~~\snd{x}~~\snd{v}~~\ldots ~~\snd{u}~~\snd{y}~~\ldots)$ or $(\ldots ~~\snd{u}~~\snd{y}~~\ldots ~~\snd{x}~~\snd{v}~~\ldots)$.
Therefore, there are at most two intervals of type 1 in $\mathcal{I}(G)$.   \qed



\end{proof}

\begin{lemma}
The maximum number of intervals of type 2 in  $\mathcal{I}(G)$ is $n$.
\label{maxType2}
\end{lemma}

\begin{proof}
First, note that for two adjacencies $\aff{x}{y}$ and $\asf{x}{z}$ in $G$ that 
do not contain marker $\circ$, if $\aff{x}{y}$ is of type 2 then $\asf{x}{z}$ 
cannot be of type 2.
Now, there is only one marker $u$ such that  $\asf{u}{\circ}$ is an adjacency 
of $G$. Let $\aff{u}{v}$ be the adjacency of $G$ having $u$ as first marker, 
then at most half of the intervals in $\mathcal{I}(G) - \{I\aff{u}{v}\}$ can 
be of type 2.
Therefore, there are at most $n$ intervals of type 2 in $\mathcal{I}(G)$.   \qed







\end{proof}

\begin{theorem}
If $\NG(G)$ contains $C$ cycles, then the BI halving distance of $G$
    is given by: 
$$d^t_{BI}(G) =  \left \lfloor \frac{n - C}{2} \right \rfloor$$
\label{th:distance}
\end{theorem}

\begin{proof}



    Since there are $2n - 1$ intervals in $\mathcal{I}(G)$, and at
    most $n+2$ are of type $1$ or $2$, then if $G$ is a genome
    containing more than three distinct markers $n~>~3$, then
    $2n-1~>~n+2$ and there exist two compatible intervals in
    $\mathcal{I}(G)$ inducing a BI operation that decreases the DCJ
    distance by $2$.

Next, we show that if $n = 2$ or $n = 3$, then $d^t_{BI}(G)=1$ and $2
\leq d^p_{DCJ}(G) \leq 3$.
 
If $n = 2$, then the genome can be written, either as
$(\circ~\fst{a}~\fst{b}~\snd{b}~\snd{a}~ \circ)$, in which case a BI
can swap $\fst{a}$ and $\fst{b}$ to produce a tandem-duplicated
genome, or as $(\circ~\fst{a}~\snd{a}~\fst{b}~\snd{b}~\circ)$, in
which case a BI can swap $a$ and $\snd{a}~\fst{b}$ to produce a
tandem-duplicated genome.

If $n = 3$, then the genome has two double-adjacencies to be
constructed, of the form $\ass{a}{b}$, $\ass{x}{y}$, with $\aff{a}{b}$
and $\aff{x}{y}$ being two adjacencies already present in the genome
such that $\fst{b} = \fst{x}$ or $\fst{b} = \snd{x}$ and $a$ and $y$
are distinct markers. One can rewrite $\aff{a}{b}$ and $\aff{x}{y}$ as
single markers since they will not be splitted, which makes a genome
with 4 markers such that at most 2 are misplaced. Then, a single BI
can produce a tandem-duplicated genome.

Now, it is easy to see to see that if  $n = 2$ or $n = 3$, then $d^p_{DCJ}(G) = n- C \leq 3$. Finally, if  $n = 2$ or $n = 3$, then $d^p_{DCJ}(G) \geq 2$, otherwise we would have  $d^p_{DCJ}(G) = 1$ which would imply, as $G$ consists in a single linear chromosome, $d^t_{BI}(G) = 0$.
In conclusion, if $n~>~3$ then there exist two compatible intervals in  $\mathcal{I}(G)$, otherwise if  $n = 2$ or $n = 3$, then $d^t_{BI}(G)=1$ and  $2 \leq d^p_{DCJ}(G) \leq 3$. Therefore $d^t_{BI} = \lfloor \frac{d^p_{DCJ}}{2}\rfloor = \lfloor \frac{n-C}{2}\rfloor$. \qed

\end{proof}





\section{Sorting algorithm}
\label{sec:scenario}



In Section \ref{sec:dist}, we showed that if a genome $G$ contains more than 
three distinct markers after reduction then there exist two compatible 
intervals in $\mathcal{I}(G)$ inducing a BI to perform.
If $G$ contains two or three distinct markers then the BI to perform can be 
trivially computed.
Thus the main concern of this section is to describe an efficient algorithm 
for finding compatible intervals when $n~>~3$.


As in Section \ref{sec:dist}, in the following, $G$ denotes a genome consisting 
of $n$ distinct markers after reduction. 
%
%
%
%
It is easy to show that the set of intervals 
$\mathcal{I}(G)$ can be built in $O(n)$ time and space complexity.



We now show that finding 2 compatible intervals in  $\mathcal{I}(G)$ can be done in $O(n)$ time and space complexity.

\begin{property}
If $n~>~3$ 
, then all the smallest intervals in $\mathcal{I}(G)$ that are not of type 2 admit compatible intervals.
\label{smallestOK}
\end{property}

\begin{proof}











    Let $J$ be a smallest interval that is not of type 2 in
    $\mathcal{I}(G)$. As $J$ is not of type 2, then $J$ has compatible
    intervals if $J$ is not of type 1.

Let us suppose that $J$ is of type 1, then for any adjacency $(a ~ b)$ such 
that markers $a$ and $b$ are not in $J$,  $\snd{a}$ and $\snd{b}$ are in $J$, 
and then $I(a ~ b)$ is strictly included in $J$ and $I(a ~ b)$ can't be of 
type 2. Such adjacency does exist as there are $n~>~3$ markers not included in $J$.
Therefore $J$ cannot be a smallest interval that is not of type 2.
\qed
\end{proof}

We are now ready to give the algorithm for sorting a duplicated genome $G$ into a tandem-duplicated genome with $\lfloor \frac{n - C}{2} \rfloor$ BI operations.



\begin{algorithm}                      
\caption{Reconstruction of a tandem-duplicated genome}          
\label{alg1}                           
\begin{algorithmic}[1]                    

\WHILE{$G$ contains more than $3$ markers}
\STATE Construct $\mathcal{I}(G)$
\STATE Pick a smallest interval $I(a ~~ b)$ that is not of type 2 in $\mathcal{I}(G)$
\STATE Find an interval $I(x ~~ y)$ in $\mathcal{I}(G)$ compatible with $I(a ~~ b)$
\STATE Perform the \BI ~ equivalent to $\DCJ(a ~~ b)$ followed by $\DCJ(x ~~ y)$
\STATE Reduce $G$
\ENDWHILE
\IF{$G$ contains $2$ or $3$ markers}
\STATE Find the last BI operation and perform it
\ENDIF
\end{algorithmic}
\end{algorithm}



\begin{theorem}

Algorithm \ref{alg1} reconstruct a tandem-duplicated genome with a BI scenario of length $\lfloor \frac{n - C}{2} \rfloor$ in $O(n^2)$ time and space complexity.
\end{theorem}

\begin{proof}
Building $\mathcal{I}(G)$ and finding two compatible intervals can be done in $O(n)$ time and space complexity. It follows that the while loop in the algorithm can be computed in $O(n^2)$ time and space complexity.



Finding and performing the last \BI~operation when  $2\leq n\leq 3$ can be done in constant time and space complexity.

Moreover, all \BI~ operations, possibly excluding the last one, are computed as 
pairs of sorting DCJ operations, which ensures that the length of the scenario 
is $\lfloor \frac{n - C}{2} \rfloor$. \qed
\end{proof}


\section{Conclusion}

In this paper, we introduced the BI halving problem. We use the DCJ model
to simulate BI operations and we showed that it is always possible to choose two consecutive 
sorting DCJ operations such that they are equivalent to a BI operation.
We thus provide a quadratic time and space algorithm to obtain a
most parsimonious scenario as any computed BI scenario is in fact an optimal DCJ scenario.
Finally, one direction for further studies of variants of the BI halving problem is to consider multichromosomal
genomes and BI operations acting on more than one chromosome.
\bibliographystyle{plain} 
\bibliography{article1}

\end{document}